\documentclass{article}
\usepackage[utf8]{inputenc}
\usepackage{latexsym,amssymb,amsfonts,amsmath,amsthm,lineno}
\usepackage[dvips]{graphicx}
\usepackage{comment}

\newtheorem{theorem}{Theorem}

\newtheorem{corollary}[theorem]{Corollary}

\title{A variant of the discrete Burgers equation derived from the correlated random walk and its ultradiscretization}
\author{{\small 
Akiko Fukuda,$^{1}$ 
\quad
Etsuo Segawa,$^{2}$ 
\quad 
Sennosuke Watanabe$^{3}$ 
}\\ 
{\scriptsize $^1$ 
Department of Mathematical Sciences, Shibaura Institute of Technology
}\\
{\scriptsize
Saitama, Saitama, 337-8570, Japan 
}\\
{\scriptsize $^{2}$ 
Graduate School of Environment and Information Sciences, Yokohama National University
}\\
{\scriptsize 
Hodogaya, Yokohama, 240-8501, Japan
} \\
{\scriptsize $^3$ 
Faculty of Informatics, The University of Fukuchiyama
}\\
{\scriptsize 
Fukuchiyama, Kyoto, 620-0886, Japan
} }
\date{}

\begin{document}

\maketitle
\begin{small}
\par\noindent
{\bf Abstract}. 
In this paper, we show that a variant of the discrete Burgers equation can be obtained through the Cole--Hopf transformation to a generalized discrete diffusion equation corresponding to the correlated random walk, which is also known as a generalization of the well known random walk. 
By applying the technique called ultradiscretization, we obtain the generalized ultradiscrete diffusion equation, the ultradiscrete Cole--Hopf transformation and a variant of the ultradiscrete Burgers equation.
Moreover, we show that the resulting ultradiscrete Burgers equation yields cellular automata which can be interpreted as a traffic flow model.


\end{small}

\section{Introduction}
It is well known that the Burgers equation can be obtained from the continuous diffusion equation by a variable transformation called the Cole--Hopf transformation \cite{Cole1951,Hopf}.
Moreover, it is shown in \cite {Hirota1979} that for the discrete diffusion equation
\begin{align}
f_{j}^{n+1} = \dfrac{1}{2}(f_{j+1}^{n} + f_{j-1}^{n}),\label{d_heat}
\end{align}
the discrete Cole--Hopf transformation
\begin{align}
u_{j}^{n} = \dfrac{f_{j+1}^{n}}{f_{j}^{n}}\label{d_ch}
\end{align}
yields the discrete Burgers equation
\begin{align}
u_{j}^{n+1} = u_{j}^{n}\dfrac{u_{j+1}^{n} + 1/u_{j}^{n}}{u_{j}^{n} + 1/u_{j-1}^{n}}.
\label{d_burgers}
\end{align}
Here, the discrete diffusion equation \eqref{d_heat} is equivalent to the evolution equation of the isotropic random walk.
Moreover, it is also known that a similar relationship can be obtained in \cite{Nishinari1998} by the ultradiscretization. 
The ultradiscrete diffusion equation 
\begin{align}
F_{j}^{n+1} = \max \left(F_{j-1}^{n},  F_{j+1}^{n}\right)\label{ud_heat}
\end{align}
transforms into the ultradiscrete Burgers equation
\begin{align}
U_{j}^{n+1}= U_{j}^{n} + \min \left(  U_{j-1}^{n},L-U_{j}^{n},  \right)  - \min \left(  U_{j}^{n}, L-U_{j+1}^{n}\right)\label{ud_burgers}
\end{align}
through the ultradiscrete Cole--Hopf transformation.
The ultradiscrete Burgers equation coincides with the elementary cellular automata of rule 184, which is regarded as a simple traffic flow model \cite{Nishinari1998}.
Originating from this model, 
some generalizations, interpretation as traffic flow, and their mathematical analyses are studied in \cite{Emmerich1998,Fukui2002}.
\par
The correlated random walk is a generalization of the random walk which depends 
on the previous position of a particle~\cite{Konno2009,Laing2001,Renshaw1981}.  
This model can be converted to a walk model in which the walker has a two-dimensional internal state 
 $\mathbb{C}^2$ whose bases are represented by $|L\rangle =[1,0]^\top$ and $|R\rangle=[0,1]^\top$, 
and the moving weights to the left and right neighbors 
are given by two-dimensional matrices $P$ and $Q$ so that $P+Q$ is a transition matrix which preserves the $\ell^1$ norm. 
A quantum walk~\cite{Kempe2003,Konno2018,Portugal2018} is known as a quantum analogue of the correlated random walk in that 
the matrix $P+Q$ is changed to a unitary matrix which preserves the $\ell^2$-norm.
In our previous work~\cite{LAA}, we introduced 
an analogue of the quantum walk over the max-plus algebra, which has a close relationship with the ultradiscretization, 
and found that the Frobenius norm on the max-plus algebra is a conservation quantity.
\par
In this paper, 
we focus on the correlated random walk and discuss the related equations and cellular automata.
We here refer to the evolution equation of the correlated random walk as the discrete correlated diffusion (dc-diffusion) equation.
We first show that by introducing a kind of discrete Cole--Hopf transformation, the dc-diffusion equation is transformed into an extension of the discrete Burgers equation, which we call the discrete correlated Burgers (dc-Burgers) equation. 
Secondly, by applying the ultradiscretization to the dc-diffusion equation and the dc-Burgers equation, we obtain the ultradiscrete correlated diffusion (udc-diffusion) equation and ultradiscrete correlated Burgers (udc-Burgers) equation, respectively. 
Moreover, we apply the ultradiscretization to the discrete Cole--Hopf transformation and show that the udc-Burgers equation is transformed into the udc-diffusion equation. 
We finally show that the resulting udc-Burgers equation can be regarded as a traffic flow model of cellular automata. 
\par
The rest of the paper is organized as follows.
In Section 2, we derive the dc-diffusion equation, discrete Cole--Hopf transformation and dc-Burgers equation.
In Section 3, we discuss the ultradiscretization of the dc-diffusion equation, discrete Cole--Hopf transformation and dc-Burgers equation.
In Section 4, we show that the udc-Burgers equation can be regarded as the cellular automata.
In Section 5, we show that the cellular automata shown in Section 4 can be regarded as the traffic flow.
The fundamental diagram of the traffic flow model is also discussed.
Finally in Section 6, concluding remarks are given.

\section{Discrete diffusion equation induced by a correlated random walk and its discrete Cole--Hopf transformation}
Let us consider the random walk model on the one-dimensional lattice in which 
a particle moves left and right neighbors with probabilities $1/2$ and $1/2$ respectively. 
The discrete diffusion equation \eqref{d_heat} is nothing but 
the time evolution of this random walk; that is, 
$f_j^n$ corresponds to the finding probability at discrete time $n$ and position $j$. 

Now let us consider the correlated random walk on the one-dimensional lattice which is a generalization of the above random walk. 
In the correlated random walk, the probabilities of a particle at each position associated with moving left and right 
are determined by the directions that a particle comes into the position. 
Then let $\mu_L^n(j)$ and $\mu_R^n(j)$ be the probabilities that a particle moves to the position $j$ from the positions $j+1$ and $j-1$ at discrete time $n$, respectively. 
The time evolution of the correlated random walk with parameters $p,q\in [0,1]$ is denoted by
\begin{align}
\begin{cases}
\mu_{L}^{n+1}(j) = p\mu_{L}^{n}(j+1)
 + (1-q) \mu_{R}^{n}(j+1),\\
\mu_{R}^{n+1}(j) = (1-p)\mu_{L}^{n}(j-1) + q \mu_{R}^{n}(j-1).\label{crw0}
\end{cases}
\end{align}
Here, the parameter $p$ (resp. $q$) is the probability that the leftward (resp. rightward) movement is repeated two times in a row, respectively. 
Note that if $p=q=1/2$, then the discrete diffusion equation~\eqref{d_heat} is recovered. 
\par
In the following, we set $p=q$ for simplicity. 
Then the second equation of \eqref{crw0} is reduced to 
\begin{align*}
\mu_{L}^{n}(j-1) = \dfrac{1}{1-p}(\mu_{R}^{n+1}(j) - p\mu_{R}^{n}(j-1)). 
\end{align*}
Inserting this into the first equation of \eqref{crw0}, we have the following recurrence equation of $\mu_R^{(n)}(j)$:
\begin{align}
&\dfrac{1}{1-p}(\mu_{R}^{n+2}(j+1) - p\mu_{R}^{n+1}(j))\nonumber\\
&\qquad  = \dfrac{p}{1-p}(\mu_{R}^{n+1}(j+2) - p\mu_{R}^{n}(j+1))
 + (1-p)\mu_{R}^{n}(j+1). \label{eq:mu_R}
\end{align}
Let us put $f_{j}^{n}:=\mu_{R}^{n}(j)$. 
The equation \eqref{eq:mu_R} implies 
\begin{align*}
f_{j+1}^{n+2} = pf_{j}^{n+1} +  pf_{j+2}^{n+1} + \{p(1-p)^{2}- p^2\} f_{j+1}^{n}
\end{align*}
Then we obtain the following three-term recurrence equation representing the correlated random walk: 
\begin{align}
f_{j}^{n+1} = p(f_{j-1}^{n} +  f_{j+1}^{n}) -(2p-1) f_{j}^{n-1}.\label{crw3}
\end{align}
If $p=1/2$, the time evolution of the usual random walk \eqref{d_heat} appears.
We call \eqref{crw3} the dc-diffusion equation. 

\par
Next, let us apply a discrete Cole--Hopf transformation to the time evolution equation of the correlated random walk \eqref{crw3}. 
Through a consideration of the usual discrete Cole--Hopf transformation \eqref{d_ch} of the random walk \eqref{d_heat},
we define another discrete Cole--Hopf transformation of the correlated random walk 
by the following pair of the spatial and temporal ratios of the finding probabilities: 
\begin{align}
u_{j}^{n} = \dfrac{f_{j+1}^{n}}{f_{j}^{n}},\quad 
v_{j}^{n} = \dfrac{f_{j}^{n+1}}{f_{j}^{n}}.\label{d_gch}
\end{align}
Then, inserting \eqref{crw3} into \eqref{d_gch}, we see that $u_j^{n}$ and $v_j^{n}$ must satisfy the following recursions:
\begin{align}
\begin{cases}
u_{j}^{n+1} =u_{j}^{n}\cdot \dfrac{p(u_{j+1}^{n} + 1/u_{j}^{n}) -(2p-1) \cdot 1/v_{j+1}^{n-1}}{p(u_{j}^{n} + 1/u_{j-1}^{n}) -(2p-1) \cdot 1/v_{j}^{n-1}},\\
v_{j}^{n+1}=v_{j}^{n}\cdot\dfrac{p(u_{j}^{n+1} + 1/u_{j-1}^{n+1}) -(2p-1)\cdot 1/v_{j}^{n}}{p(u_{j}^{n} + 1/u_{j-1}^{n} ) -(2p-1) \cdot 1/v_{j}^{n-1}}.\label{d_gburgers}
\end{cases}
\end{align}
We call (\ref{d_gburgers}) the dc-Burgers equation. 
Note that if we put $p=1/2$ in \eqref{d_gburgers}, the first equation \eqref{d_gburgers} coincides with the discrete Burgers equation \eqref{d_burgers}.

\section{
Ultradiscretization of the dc-diffusion equation and its Cole--Hopf transformation
}
In this section we will give an ultradiscretization of the dc-diffusion equation \eqref{crw3} and its Cole--Hopf transformation.

Let a variable transformation be
\begin{align}
f_{j}^{n} = p^{n}e^{F_{j}^{n}/\varepsilon}.\label{crw:ud_vt}
\end{align}
By applying \eqref{crw:ud_vt} to \eqref{crw3},  we obtain
\begin{align*}
e^{F_{j}^{n+1}/\varepsilon} = e^{F_{j-1}^{n}/\varepsilon} +  e^{F_{j+1}^{n}/\varepsilon} +\dfrac{1-2p}{p^2} e^{F_{j}^{n-1}/\varepsilon}.
\end{align*}
Here, we assume that $\frac{1-2p}{p^2} \geq 0$, that is, $0\leq p \leq1/2$, and let
\begin{align*}
\dfrac{1-2p}{p^2}=e^{R/\varepsilon},
\end{align*}
then we have
\begin{align*}
e^{F_{j}^{n+1}/\varepsilon} = e^{F_{j-1}^{n}/\varepsilon} +  e^{F_{j+1}^{n}/\varepsilon} +e^{R/\varepsilon} e^{F_{j}^{n-1}/\varepsilon}.
\end{align*}
Taking the $\log$ and multiplying both sides of the equation by $\varepsilon$, we have
\begin{align*}
\varepsilon\log e^{F_{j}^{n+1}/\varepsilon} = \varepsilon\log \left(e^{F_{j-1}^{n}/\varepsilon} +  e^{F_{j+1}^{n}/\varepsilon} +e^{R/\varepsilon} e^{F_{j}^{n-1}/\varepsilon}\right).
\end{align*}
Taking a limit $\varepsilon\to +0$ on the above equation, we obtain
\begin{align}
F_{j}^{n+1} = \max \left(F_{j-1}^{n},  F_{j+1}^{n}, R +F_{j}^{n-1}\right),\quad R\geq -\infty. \label{ud_gheat}
\end{align}
We call \eqref{ud_gheat} the udc-diffusion equation.
Here, if $R=-\infty$, then \eqref{ud_gheat} coincides with 
the ultradiscrete diffusion equation \eqref{ud_heat}.
\par
Next, we will apply the ultradiscretization to the discrete Cole--Hopf transformation \eqref{d_gch}.
Let us introduce new variable transformations as
\begin{align}
u_{j}^{n} = e^{(U_{j}^{n}-L/2)/\varepsilon},\quad
v_{j}^{n} = e^{(V_{j}^{n}-L/2)/\varepsilon}.\label{crw:ud_vt2}
\end{align}
By applying the variable transformation \eqref{crw:ud_vt2} into \eqref{d_gch}, we have
\begin{align*}
e^{(U_{j}^{n}-L/2)/\varepsilon} = \dfrac{e^{F_{j+1}^{n}/\varepsilon}}{e^{F_{j}^{n}/\varepsilon}},\quad
e^{(V_{j}^{n}-L/2)/\varepsilon} = \dfrac{e^{F_{j}^{n+1}/\varepsilon}}{e^{F_{j}^{n}/\varepsilon}}
.
\end{align*}
Taking the ultradiscrete limit, it holds that
\begin{align}
\begin{cases}
U_{j}^{n} = F_{j+1}^{n} - F_{j}^{n} +L/2,\\
V_{j}^{n} = F_{j}^{n+1} - F_{j}^{n} + L/2.\label{ud_gch}
\end{cases}
\end{align}
\par
Next, we apply the ultradiscrete Cole--Hopf transformation \eqref{ud_gch} to the udc-diffusion equation \eqref{ud_gheat}.
For the 1st equation of \eqref{ud_gch} at discrete time $n+1$, 
substitute \eqref{ud_gheat} for $F_{j}^{n+1}$
\begin{align*}
U_{j}^{n+1} &= \max \left(F_{j}^{n},  F_{j+2}^{n}, R +F_{j+1}^{n-1}\right) - \max \left(F_{j-1}^{n},  F_{j+1}^{n}, R +F_{j}^{n-1}\right) +L/2\\
&= (F_{j+1}^{n} -F_{j}^{n} +L/2) \notag\\
&\qquad+ \max \left(F_{j}^{n} -F_{j+1}^{n},  F_{j+2}^{n}-F_{j+1}^{n}, R +F_{j+1}^{n-1}-F_{j+1}^{n}\right)\notag \\
&\qquad - \max \left(F_{j-1}^{n}-F_{j}^{n}, F_{j+1}^{n}-F_{j}^{n}, R +F_{j}^{n-1}-F_{j}^{n}\right)\\
&= U_{j}^{n}+ \max \left( -U_{j}^{n}+L/2,  U_{j+1}^{n}-L/2, R -V_{j+1}^{n-1}+L/2\right) \notag\\
&\qquad - \max \left( -U_{j-1}^{n}+L/2,  U_{j}^{n}-L/2, R -V_{j}^{n-1}+L/2\right)\\
&= U_{j}^{n} + \max \left( -U_{j}^{n},  U_{j+1}^{n}-L, R -V_{j+1}^{n-1}\right) - \max \left( -U_{j-1}^{n},  U_{j}^{n}-L, R -V_{j}^{n-1}\right)\\
&= U_{j}^{n} + \min \left( U_{j-1}^{n},  L-U_{j}^{n}, V_{j}^{n-1}-R\right)- \min \left( U_{j}^{n},  L-U_{j+1}^{n},  V_{j+1}^{n-1}-R\right). 
\end{align*}
In the same way, the 2nd equation of \eqref{ud_gch} at discrete time $n+1$ is transformed into
\begin{align*}
V_{j}^{n+1} 
= V_{j}^{n}+ \min \left( U_{j-1}^{n},  L-U_{j}^{n},  V_{j}^{n-1}-R\right)-\min \left( U_{j-1}^{n+1},  L-U_{j}^{n+1}, V_{j}^{n}-R\right). 
\end{align*}
Therefore, the ultradiscrete Cole--Hopf transformation \eqref{ud_gch} of  udc-diffusion equation \eqref{ud_gheat} yields
\begin{align}
\begin{cases}
U_{j}^{n+1}= U_{j}^{n} + \min \left(  U_{j-1}^{n},L-U_{j}^{n},  V_{j}^{n-1}-R\right)  - \min \left(  U_{j}^{n}, L-U_{j+1}^{n},   V_{j+1}^{n-1}-R\right),\\
V_{j}^{n+1}= V_{j}^{n} + \min \left(   U_{j-1}^{n},L-U_{j}^{n},   V_{j}^{n-1}-R\right)-\min \left(  U_{j-1}^{n+1}, L-U_{j}^{n+1},  V_{j}^{n}-R\right).\label{ud_gburgers}
\end{cases}
\end{align}
We call \eqref{ud_gburgers} the ultradiscrete correlated Burgers (udc-Burgers) equation.
In the case that $R=-\infty$, the 1st equation of \eqref{ud_gburgers} coincides with \eqref{ud_burgers}. 
\par
Next we apply the ultradiscretization into the dc-Burgers equation \eqref{d_gburgers}. 
Let us recall \eqref{crw:ud_vt2} and introduce the new variable transformation
\begin{align*}
\dfrac{1-2p}{p}=e^{R/\varepsilon}.
\end{align*}
Then the 1st equation of \eqref{d_gburgers} yields
\begin{align*}
&e^{(U_{j}^{n+1}-L/2)/\varepsilon} \notag\\
&\quad =e^{(U_{j}^{n}-L/2)/\varepsilon}\cdot \dfrac{(e^{(U_{j+1}^{n}-L/2)/\varepsilon}+1/e^{(U_{j}^{n}-L/2)/\varepsilon}) +e^{R/\varepsilon} \cdot 1/e^{(V_{j+1}^{n-1}-L/2)/\varepsilon}}{(e^{(U_{j}^{n}-L/2)/\varepsilon}+1/e^{(U_{j-1}^{n}-L/2)/\varepsilon}) +e^{R/\varepsilon}\cdot 1/e^{(V_{j}^{n-1}-L/2)/\varepsilon}}.
\end{align*}
Taking the ultradiscrete limit into the above equation, we obtain
\begin{align*}
U_{j}^{n+1}=U_{j}^{n}&+
\min( L-U_{j}^{n},~U_{j-1}^{n},~V_{j}^{n-1}-R)\\
&-\min (L-U_{j+1}^{n}, ~U_{j}^{n},~ V_{j+1}^{n-1}-R).
\end{align*}
Similarly, the 2nd equation of \eqref{d_gburgers} yields 
\begin{align*}
V_{j}^{n+1}= V_{j}^{n}&-\min ( L-U_{j}^{n+1},~ U_{j-1}^{n+1},~ V_{j}^{n}-R)\\
&+\min ( L-U_{j}^{n},~ U_{j-1}^{n},~ V_{j}^{n-1}-R).
\end{align*}
We can see that the ultradiscrete limit
\begin{align*}
\begin{cases}
U_{j}^{n+1}  =U_{j}^{n}+
\min( U_{j-1}^{n},~L-U_{j}^{n},~V_{j}^{n-1}-R)
-\min ( U_{j}^{n},~L-U_{j+1}^{n},~ V_{j+1}^{n-1}-R),\\
V_{j}^{n+1}= V_{j}^{n}+\min ( U_{j-1}^{n},~L-U_{j}^{n},~ V_{j}^{n-1}-R)-\min ( U_{j-1}^{n+1},~L-U_{j}^{n+1},~ V_{j}^{n}-R)
\end{cases}
\end{align*}
of the dc-Burgers equation \eqref{d_gburgers}
coincides with \eqref{ud_gburgers} which is derived through the Cole--Hopf transformation \eqref{ud_gch} of the udc-diffusion equation \eqref{ud_gheat}.

\section{The ultradiscrete correlated Burgers equation and cellular automata}
In the udc-Burgers equation \eqref{ud_gburgers}, 
let $R=-\infty$, then the 1st equation,
\begin{align*}
U_{j}^{n+1}  =U_{j}^{n}+
\min( L-U_{j}^{n},~U_{j-1}^{n})
-\min (L-U_{j+1}^{n}, ~U_{j}^{n}),
\end{align*}
coincides with the ultradiscrete Burgers equation \eqref{ud_burgers},
which is equivalent to the elementary cellular automata of rule 184
if we set $L=1$ and the initial values are in $\{0,1\}$. 
Let $\tilde{V}_{j}^{n}:=V_{j}^{n}-R$. Then the udc-Burgers equation \eqref{ud_gburgers} transforms into
\begin{align}
\begin{cases}
U_{j}^{n+1}  =U_{j}^{n}+
\min( U_{j-1}^{n},~L-U_{j}^{n},~\tilde{V}_{j}^{n-1})
-\min ( U_{j}^{n},~L-U_{j+1}^{n},~ \tilde{V}_{j+1}^{n-1}),\\
\tilde{V}_{j}^{n+1}= \tilde{V}_{j}^{n}+\min (  U_{j-1}^{n},~L-U_{j}^{n},~ \tilde{V}_{j}^{n-1})-\min (  U_{j-1}^{n+1},~L-U_{j}^{n+1}, \tilde{V}_{j}^{n}),
\end{cases}
\label{udpb1}
\end{align}
in which the parameter $R$ is deleted.
For the udc-Burgers equation \eqref{ud_gburgers},
let $X_{j}^{n}=\min( U_{j-1}^{n},~L-U_{j}^{n},~\tilde{V}_{j}^{n-1})$;
then \eqref{udpb1} can be rewritten as
\begin{align}
\begin{cases}
U_{j}^{n+1}  =U_{j}^{n}+
X_{j}^{n}-X_{j+1}^{n},\\
\tilde{V}_{j}^{n+1}= \tilde{V}_{j}^{n}+X_{j}^{n}-X_{j}^{n+1}.
\end{cases}\label{udpb2}
\end{align}
Here we assume that the sites $j$ satisfy the periodic boundary condition with integer period $N$,
i.e., $U_{j}^{n}=U_{j+N}^{n}, \tilde{V}_{j}^{n}=\tilde{V}_{j+N}^{n}$.
In the following Theorem \ref{thm_CA}, we discuss the initial condition that the values of $U_{j}^{n}$ and $\tilde{V}_{j}^{n}$ for any $j$ are in the closed interval independent of $n$.
\begin{theorem}\label{thm_CA}
In the udc-Burgers equation \eqref{udpb1},
we assume that
\begin{align}
0\leq U_{j}^{0} \leq L,\quad
0\leq \tilde{V}_{j}^{-1} + \tilde{V}_{j}^{0}\leq L.\label{a1}
\end{align}
Then, for any discrete time $n\geq 1$, it holds that 
\begin{align}
0\leq U_{j}^{n} \leq L,\quad
0\leq \tilde{V}_{j}^{n}\leq L.\label{eq:thm1}
\end{align}
\end{theorem}
\begin{proof}
This proof is given by induction with respect to $n$.
We assume that \eqref{eq:thm1} holds true for the discrete time less than or equal to $n$.
In the 1st equation of \eqref{udpb1}, we have
\begin{align}
U_{j}^{n+1}  &=U_{j}^{n}+
\min( U_{j-1}^{n},~L-U_{j}^{n},~\tilde{V}_{j}^{n-1})
-\min (U_{j}^{n},~L-U_{j+1}^{n}, ~ \tilde{V}_{j+1}^{n-1})\notag\\
&=
\min( U_{j-1}^{n}+U_{j}^{n},~L,~\tilde{V}_{j}^{n-1}+U_{j}^{n})
-\min (U_{j}^{n},~L-U_{j+1}^{n}, ~ \tilde{V}_{j+1}^{n-1}).\label{13.13}
\end{align}
The 1st term of \eqref{13.13} yields
\begin{align*}
\min( U_{j-1}^{n}+U_{j}^{n},~L,~\tilde{V}_{j}^{n-1}+U_{j}^{n})
\leq L,
\end{align*}
and, from the assumption of the induction, the 2nd term holds
\begin{align*}
\min (U_{j}^{n},~L-U_{j+1}^{n}, ~ \tilde{V}_{j+1}^{n-1})\geq 0.
\end{align*}
Therefore we have $U_{j}^{n+1}\leq L$.
On the other hand, the 1st equation of \eqref{udpb1} can also be written as
\begin{align*}
U_{j}^{n+1}  &=
\min( U_{j-1}^{n},~L-U_{j}^{n},~\tilde{V}_{j}^{n-1})
-\min (0,~L-U_{j+1}^{n}-U_{j}^{n}, ~ \tilde{V}_{j+1}^{n-1}-U_{j}^{n}).
\end{align*}
Here, the 1st term of the right-hand side satisfies
\begin{align*}
\min( U_{j-1}^{n},~L-U_{j}^{n},~\tilde{V}_{j}^{n-1})\geq 0,
\end{align*}
and the 2nd term satisfies
\begin{align*}
\min (0,~L-U_{j+1}^{n}-U_{j}^{n}, ~ \tilde{V}_{j+1}^{n-1}-U_{j}^{n})\leq 0.
\end{align*}
Hence we have $U_{j}^{n+1}\geq 0$.
Therefore we obtain $0\leq U_{j}^{n+1}\leq L$. 
\par
In the 2nd equation of \eqref{udpb1}, we have
\begin{align}
\tilde{V}_{j}^{n+1}&= \tilde{V}_{j}^{n}+\min ( U_{j-1}^{n},~L-U_{j}^{n},~  \tilde{V}_{j}^{n-1})-\min (  U_{j-1}^{n+1},~L-U_{j}^{n+1},~ \tilde{V}_{j}^{n})\notag\\
&= \min (  U_{j-1}^{n},~L-U_{j}^{n},~ \tilde{V}_{j}^{n-1})-\min ( U_{j-1}^{n+1}-\tilde{V}_{j}^{n},~L-U_{j}^{n+1}-\tilde{V}_{j}^{n},~  0).\label{eq:thm1:V}
\end{align}
Here, from the assumption of the induction, the 1st term of \eqref{eq:thm1:V} holds
\begin{align*}
\min ( U_{j-1}^{n},~ L-U_{j}^{n},~ \tilde{V}_{j}^{n-1})\geq 0,
\end{align*}
and the 2nd term holds
\begin{align*}
\min (  U_{j-1}^{n+1}-\tilde{V}_{j}^{n},~L-U_{j}^{n+1}-\tilde{V}_{j}^{n},~ 0)\leq 0.
\end{align*}
Hence we have $\tilde{V}_{j}^{n+1}\geq 0$.
In the 2nd equation of \eqref{udpb2},
we take the summation from $0$ to $n$ on both sides of the equation, and then we have
\begin{align*}
\sum_{p=0}^{n}
(\tilde{V}_{j}^{p+1}-\tilde{V}_{j}^{p})
=\sum_{p=0}^{n}(X_{j}^{p}-X_{j}^{p+1}).
\end{align*}
Here, the left- and right-hand sides are $-\tilde{V}_{j}^{0} + \tilde{V}_{j}^{n+1}$ and $X_{j}^{0}-X_{j}^{n+1}$, respectively.
Therefore we have
\begin{align*}
\tilde{V}_{j}^{n+1} &= X_{j}^{0}-X_{j}^{n+1}+\tilde{V}_{j}^{0}\\
&=\min (  U_{j-1}^{0},~L-U_{j}^{0},~ \tilde{V}_{j}^{-1})
-\min (  U_{j-1}^{n+1},~L-U_{j}^{n+1},~ \tilde{V}_{j}^{n})
+\tilde{V}_{j}^{0}\notag\\
&=\min (  U_{j-1}^{0}+\tilde{V}_{j}^{0}
,~L-U_{j}^{0}+\tilde{V}_{j}^{0}
,~ \tilde{V}_{j}^{-1}+\tilde{V}_{j}^{0}
)
-\min ( U_{j-1}^{n+1},~  L-U_{j}^{n+1},~\tilde{V}_{j}^{n}).
\end{align*}
From the assumption \eqref{a1}, the 1st term of the above equation holds 
\begin{align*}
\min (  U_{j-1}^{0}+\tilde{V}_{j}^{0}
,~L-U_{j}^{0}+\tilde{V}_{j}^{0}
,~ \tilde{V}_{j}^{-1}+\tilde{V}_{j}^{0}
)\leq L
\end{align*}
and $0\leq U_{j}^{n+1}\leq L$, which we already proved, and the 2nd term holds
\begin{align*}
\min (  U_{j-1}^{n+1},~L-U_{j}^{n+1},~ \tilde{V}_{j}^{n})\geq 0.
\end{align*}
So we have $\tilde{V}_{j}^{n+1}\leq L$.
Therefore we have $0\leq\tilde{V}_{j}^{n+1}\leq L$.
\end{proof}
From the Theorem \ref{thm_CA}, we see that the udc-Burgers equation \eqref{udpb1} can be regarded as the cellular automata with state $U_{j}^{n}\in\{0,1,\dots,L\}$ and $\tilde{V}_{j}^{n}\in\{0,1,\dots,L\}$.
From the 2nd equation of \eqref{udpb2}, we have
\begin{align}
\tilde{V}_{j}^{n}
=-X_{j}^{n}+\min( U_{j-1}^{0},~L-U_{j}^{0},~\tilde{V}_{j}^{-1})+\tilde{V}_{j}^{0}.\label{tildeV}
\end{align}
Here, we introduce
\begin{align*}
I_{j}^{0}:=\min( U_{j-1}^{0},~L-U_{j}^{0},~\tilde{V}_{j}^{-1})+\tilde{V}_{j}^{0},
\end{align*}
which is determined by the initial values of the udc-Burgers equation.
Hence we have the following theorem.
\begin{theorem}\label{thm_I}
In the udc-Burgers equation \eqref{udpb1},
let us assume that $0\leq U_{j}^{0}\leq L$, $ \tilde{V}_{j}^{-1}\geq 0$ and $\tilde{V}_{j}^{0}\geq 0$.
Then, for any discrete time $n\geq 1$, it holds that 
\begin{align}
0\leq U_{j}^{n}\leq L,\quad 
0\leq \tilde{V}_{j}^{n}\leq I_{j}^{0}.\label{th3eq}
\end{align}
\end{theorem}
\begin{proof}
The first equation of \eqref{th3eq} is already proved in Theorem \ref{thm_CA}.
For the 2nd equation of \eqref{th3eq}, 
from \eqref{tildeV}, we easily have
\begin{align}
\tilde{V}_{j}^{n}=-X_{j}^{n}+I_{j}^{0}.\label{V=-X+I}
\end{align}
From $X_{j}^{n}\geq 0$, we have $\tilde{V}_{j}^{n}\leq I_{j}^{0}$.
\end{proof}

\begin{corollary}\label{Cor}
Let $M:=\displaystyle\max_{j\in \mathbb{Z}}I_{j}^{0}<\infty$; then the udc-Burgers equation \eqref{udpb1} yields cellular automata of 
$(U_{j}^{n},\tilde{V}_{j}^{n})\in \{0,1,\dots,L\}\times \{0,1,\dots, M\}$.
\end{corollary}
\par
Fuzzy cellular automata are known as cellular automata that extend the range of elementary cellular automata from $\{0,1\}$ to $[0,1]\subset \mathbb{R}$ \cite{floccini}.
It is known that the straightforward fuzzification of generalized cellular automata whose range is $\{0,1,\dots,L\}$ is difficult. 
Recently, a skillful fuzzification technique for generalized cellular automata was introduced by some of authors in \cite{JCA}.
However, in the proof of Theorem \ref{thm_CA},  
$U_{j}^{n}$ and $\tilde{V}_{j}^{n}$ are not assumed to be integers. 
Therefore, it is remarkable to note that the range 
$(U_{j}^{n},\tilde{V}_{j}^{n})\in \{0,1,\dots,L\}\times \{0,1,\dots, M\}$
can be directly generalized into $(U_{j}^{n},\tilde{V}_{j}^{n})\in [0,L] \times [0,M]$ in Corollary \ref{Cor}.

\section{The correlated Burgers cellular automata}
Let us refer to the cellular automata derived from the udc-Burgers equation \eqref{udpb1} as the correlated Burgers cellular automata (c-BCA).
In this section, we show that the c-BCA can be regarded as a kind of traffic flow model.
We also discuss the fundamental diagram of the traffic flow model.
\subsection{Interpretation of the c-BCA as a traffic flow model}
Let $U_{j}^{n}$ be the number of cars at site $j$ and discrete time $n$ and 
let $\tilde{V}_{j}^{n-1}$ denote the maximum inflow, namely, the upper limit of the number of cars that can move from site $j-1$ to $j$ in the transition from discrete time $n$ to $n+1$.
We here assume that each site contains at most $L$ cars.
Then, $X_{j}^{n} = \min(U_{j-1}^{n},~L-U_{j}^{n},~\tilde{V}_{j}^{n-1})$ represents the inflow toward site $j$ at discrete time $n$.
Therefore the 1st equation of the udc-Burgers equation \eqref{udpb2} represents the update of the number of cars as follows.
\begin{align*}
\underbrace{U_{j}^{n+1}}_{\substack{{\rm number~of ~cars}\\{\rm at~discrete~time}~n+1~\\{\rm and~site}~j}}
=
\underbrace{U_{j}^{n}}_{\substack{{\rm number~of ~cars}\\{\rm at~discrete~time}~n~\\{\rm and~site}~j}}
+
\underbrace{X_{j}^{n}}_{\substack{{\rm inflow~toward~site~}j \\{\rm at~discrete~time~}n}}
-\underbrace{X_{j+1}^{n}.}_{\substack{{\rm outflow~toward~site~}j+1\\{\rm at~discrete~time~}n}}
\end{align*}
The following theorem states that the udc-Burgers equation \eqref{udpb1} follows the conservation law.
\begin{theorem}\label{thm:CQ}
Let $N$ be the number of sites and impose the periodic boundary condition on the udc-Burgers equation \eqref{udpb1}.
Then, in \eqref{udpb1}, for any discrete time $n$, it holds that 
\begin{align*}
\sum_{j=0}^{N-1} U_{j}^{n+1}  =\sum_{j=0}^{N-1}U_{j}^{n}.
\end{align*}
\end{theorem}
\begin{proof}
In the 1st equation of \eqref{udpb2}, let us take the summation from 0 to $N$.
Then we have 
\begin{align*}
\sum_{j=0}^{N-1} U_{j}^{n+1}  =\sum_{j=0}^{N-1}U_{j}^{n}+
\sum_{j=0}^{N-1}(X_{j}^{n}-X_{j+1}^{n}).
\end{align*}
From $\sum_{j=0}^{N-1}(X_{j}^{n}-X_{j+1}^{n})=X_{0}^{n}-X_{N}^{n}$ and considering the periodic boundary condition,
we have
\begin{align*}
\sum_{j=0}^{N-1} U_{j}^{n+1}  =\sum_{j=0}^{N-1}U_{j}^{n}.
\end{align*}
\end{proof}
From Theorem \ref{thm:CQ}, we see that the summation of the initial values $U_{j}^{0},~j=0,1,\dots,N-1$ are conserved quantities, which means that the total number of cars does not change with time.
\par

Note that if $p=1/2$, that is, the simple random walk case, the inflow $X_j^n$ is reduced to just $\min(U_{j-1}^n, L-U_j^n)$ without any influence of $\tilde{V}_j^n$ since $\tilde{V}_j^n=+\infty$. 
Thus the finiteness of $\tilde{V}_j^n$ leads to the characterization of ``correlated" random walk. 
Then let us consider $\tilde{V}_j^n$ more precisely in the following.
Hereinafter, let $\tilde{V}_{j}^{-1}=0$ for all $j$, then it holds that $ I_{j}^{0} = \tilde{V}_{j}^{0}$.
Namely, for all discrete time $n\geq 1$, $0<\tilde{V}_{j}^{n}\leq \tilde{V}_{j}^{0}$ holds in Theorem \ref{thm_I}.
Therefore, under the condition $\tilde{V}_{j}^{-1}=0$, we can give arbitrary initial values $\tilde{V}_{j}^{0}$, which means that, 
for each $j$, we can control the upper bounds of $\tilde{V}_{j}^{n}$ by the initial values $\tilde{V}_{j}^{0}$.
Then the 2nd equation of the udc-Burgers equation \eqref{udpb2}, which is transformed into \eqref{V=-X+I}, can be interpreted as 
\begin{align*}
\underbrace{\tilde{V}_{j}^{n}}_{\substack{{\rm maximum~inflow~of}\\{\rm transition~from~}n+1{\rm~to~}n+2\\{\rm ~at~site~}j-1{\rm ~to~}j}} 
= \underbrace{V_j^{0}}_{\substack{{\rm potential~quantity }\\{\rm ~of~movement~at~site~}j }} - \underbrace{X_{j}^{n},}_{\substack{{\rm inflow~toward~site~}j \\{\rm at~discrete~time~}n}}
\end{align*}
which means that the maximum inflow of transition at the next time (from $n+1 \to n+2$) can be determined by subtracting the amount of inflow to site $j$ at the current time ($n\to n+1$) from the initially determined inflow limit $\tilde{V}_j^{0}$ at site $j$.
That is, whereas the maximum inflow $\tilde{V}_j^{n}$ at discrete time $n$ decreases depending on the previous inflow $\tilde{X}_j^{n}$, the upper bound of $\tilde{V}_j^{n}$ is determined by $\tilde{V}_j^{0}$.
Therefore the inflow toward site $j$ at time $n$ is described by
\begin{align*}  
X_j^n = \min (U_{j-1}^n, L-U_j^n, V_j^0-X_j^{n-1} ).
\end{align*}
The  dependence on the direction at the previous time step of the correlated random walk appears as the dependence on the inflow at the previous time step $V_j^0 - X_j^{n-1}<\infty$ of this cellular automaton. 
\par
In the following, we give explanations of the movement of cars in the 
udc-Burgers equation \eqref{udpb1} in the case of $L=1$, $\tilde{V}_{j}^{-1}=0$ and $\tilde{V}_{j}^{0}\in\{0,1\}$ for all $j$.
From the 1st equation of \eqref{udpb1} and \eqref{udpb2} of $L=1$, the values of $X_{j}^{n}$, $X_{j+1}^{n}$ and resulting $U_{j}^{n+1}$
according to the combinations of the values $U_{j-1}^{n}$, $U_{j}^{n}$, $U_{j+1}^{n}$, $\tilde{V}_{j}^{n-1}$ and  $\tilde{V}_{j+1}^{n-1}$ are shown in Table \ref{table1_temp2}.
\begin{table}\caption{$X_{j}^{n}-X_{j+1}^{n}$\label{table1_temp2}}
\begin{center}
\begin{tabular}{cccc|cc|cc|c|c}
No.&$U_{j-1}^{n}$&$U_{j}^{n}$&$U_{j+1}^{n}$&$\tilde{V}_{j}^{n-1}$&$\tilde{V}_{j+1}^{n-1}$&$X_{j}^{n}$&$X_{j+1}^{n}$
&$U_{j}^{n+1}$&Cases\\ \hline
1	&	0	&	0	&	0	&	0	&	0	&	0	&	0		&0&I\\
2	&	0	&	0	&	0	&	0	&	1	&	0	&	0		&0&I\\
3	&	0	&	0	&	0	&	1	&	0	&	0	&	0		&0&I\\
4	&	0	&	0	&	0	&	1	&	1	&	0	&	0		&0&I\\
5	&	0	&	0	&	1	&	0	&	0	&	0	&	0		&0&I\\
6	&	0	&	0	&	1	&	0	&	1	&	0	&	0		&0&I\\
7	&	0	&	0	&	1	&	1	&	0	&	0	&	0		&0&I\\
8	&	0	&	0	&	1	&	1	&	1	&	0	&	0		&0&I\\
9	&	0	&	1	&	0	&	0	&	0	&	0	&	0		&1&VI\\
10	&	0	&	1	&	0	&	0	&	1	&	0	&	1	 &0&IV\\
11	&	0	&	1	&	0	&	1	&	0	&	0	&	0		&1&VI\\
12	&	0	&	1	&	0	&	1	&	1	&	0	&	1		&0&IV\\
13	&	0	&	1	&	1	&	0	&	0	&	0	&	0		&1&II\\
14	&	0	&	1	&	1	&	0	&	1	&	0	&	0		&1&II\\
15	&	0	&	1	&	1	&	1	&	0	&	0	&	0		&1&II\\
16	&	0	&	1	&	1	&	1	&	1	&	0	&	0		&1&II\\
17	&	1	&	0	&	0	&	0	&	0	&	0	&	0		&0&V\\
18	&	1	&	0	&	0	&	0	&	1	&	0	&	0		&0&V\\
19	&	1	&	0	&	0	&	1	&	0	&	1	&	0		&1&III\\
20	&	1	&	0	&	0	&	1	&	1	&	1	&	0		&1&III\\
21	&	1	&	0	&	1	&	0	&	0	&	0	&	0		&0&V\\
22	&	1	&	0	&	1	&	0	&	1	&	0	&	0		&0&V\\
23	&	1	&	0	&	1	&	1	&	0	&	1	&	0		&1&III\\
24	&	1	&	0	&	1	&	1	&	1	&	1	&	0		&1&III\\
25	&	1	&	1	&	0	&	0	&	0	&	0	&	0		&1&VI\\
26	&	1	&	1	&	0	&	0	&	1	&	0	&	1		&0&IV\\
27	&	1	&	1	&	0	&	1	&	0	&	0	&	0		&1&VI\\
28	&	1	&	1	&	0	&	1	&	1	&	0	&	1		&0&IV\\
29	&	1	&	1	&	1	&	0	&	0	&	0	&	0		&1&II\\
30	&	1	&	1	&	1	&	0	&	1	&	0	&	0		&1&II\\
31	&	1	&	1	&	1	&	1	&	0	&	0	&	0		&1&II\\
32	&	1	&	1	&	1	&	1	&	1	&	0	&	0		&1&II\\
\end{tabular}
\end{center}
\end{table}
In case I of Table \ref{table1_temp2},
since there are no cars at site $j-1$ and $j$,
there are no movements of cars
independent of the values of $\tilde{V}_{j}^{n-1}$.
In case II, 
there are cars at sites $j$ and $j+1$.
So the cars at site $j$ can not move to $j+1$ independent of the values of $\tilde{V}_{j}^{n-1}$.
In case III, 
since there is a car at site $j-1$, the site $j$ is empty, and the maximum inflow $\tilde{V}_{j}^{n-1}=1$,
the car at site $j-1$ moves to site $j$ at discrete time $n+1$.
In case IV, 
since there is a car at site $j$, the site $j+1$ is empty, and the maximum inflow $\tilde{V}_{j+1}^{n-1}=1$,
the car at site $j$ moves to site $j+1$ at discrete time $n+1$.
In case V, 
there is a car at site $j-1$ and the site $j$ is empty.
But the car at site $j-1$ can not move because $\tilde{V}_{j}^{n-1}=0$.
In case VI, 
there is a car at site $j$ and the site $j+1$ is empty.
But the car at site $j$ can not move because $\tilde{V}_{j+1}^{n-1}=0$. 
\par
In Figure \ref{fig:L=1}, we share the figures of the time evolution of $U_{j}^{n}$ and $\tilde{V}_{j}^{n}$ in the case that $L=1$. 
The horizontal and vertical axis denote site $j=0,1,\dots,29$ and discrete time $t=-1, 0,1,\dots,28$, from the top, respectively.
The white and black cells denote $0$ and $1$, respectively. 
In Figure \ref{fig:L=1}
the initial values of $\tilde{V}_{j}^{0}$ are set as
$\tilde{V}_{21}^{0}=0$
and 
$\tilde{V}_{j}^{0}=1$ for other $j$. 
For any $j$, $U_{j}^{-1}$ is set to $0$ as a matter of convenience.
The initial values of $U_{j}^{0}$ are set as 0 or 1 randomly.
At the point where $\tilde{V}_{21}^{t}=0$, we can  observe that the car $U_{20}^{t}$ can not move forward and a traffic jam occurs.

\begin{figure}[htbp]
\begin{center}
\includegraphics[scale=0.4]{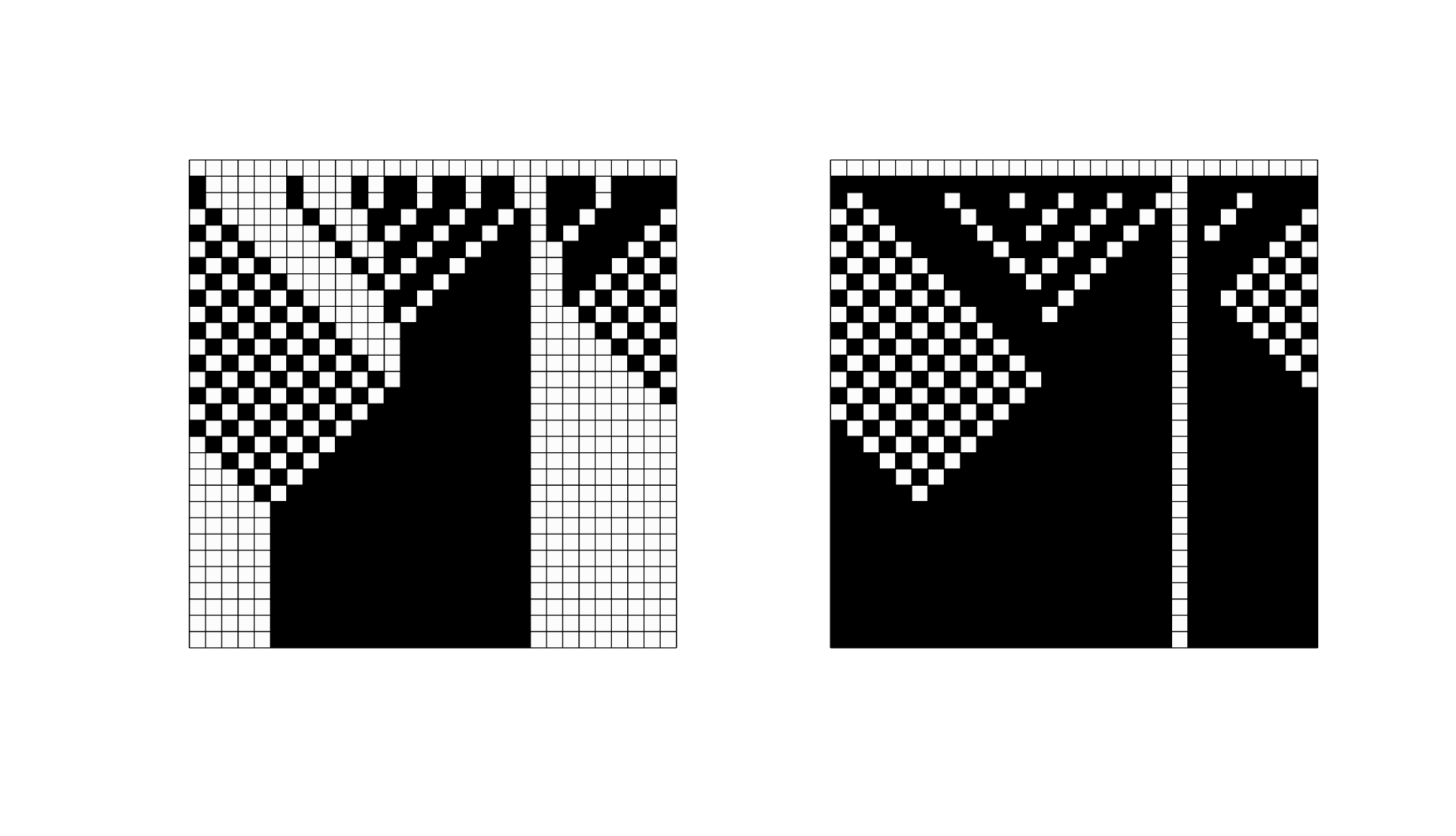}
\caption{Time evolution of $U_{j}^{n}$ (left) and $\tilde{V}_{j}^{n}$ (right) in the case that $L=1$.}\label{fig:L=1}
\end{center}
\begin{center}
\includegraphics[scale=0.4]{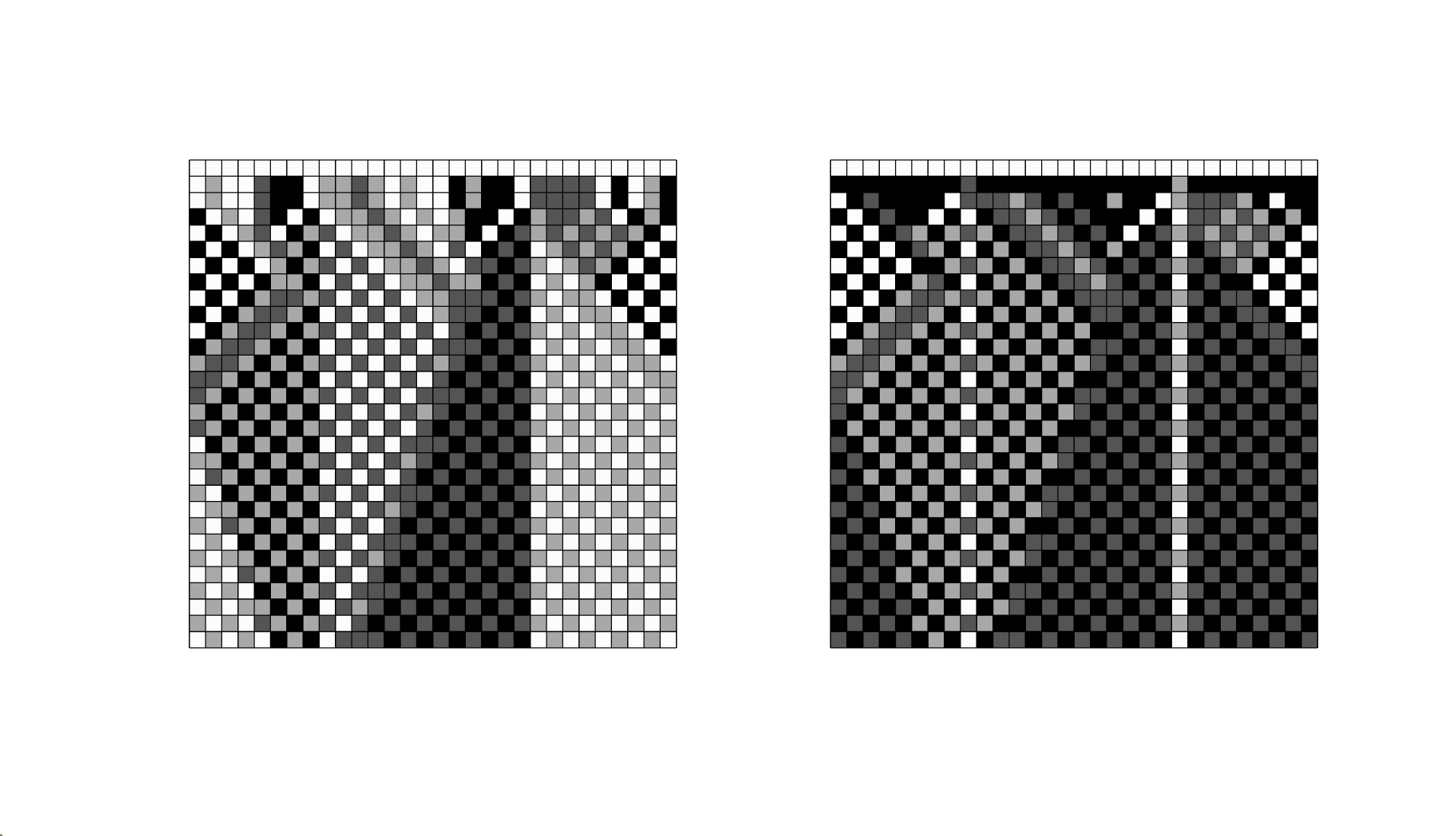}
\caption{Time evolution of $U_{j}^{n}$ (left) and $\tilde{V}_{j}^{n}$ (right) in the case that $L=3$.}\label{fig:L=3}
\end{center}
\end{figure}

In Figure \ref{fig:L=3}, we show the figure of the time evolution of $U_{j}^{n}$ and $\tilde{V}_{j}^{n}$ in the case that $L=3$. 
The white, light gray, dark gray and black cells  denote 0, 1, 2 and 3, respectively. 
The initial values of $\tilde{V}_{j}^{0}$
are set as
$\tilde{V}_{8}^{0}=2$,
$\tilde{V}_{21}^{0}=1$ and
$\tilde{V}_{j}^{0}=3$ for other $j$.
It can be observed that at most two cars are allowed to move forward per unit of time at the point where $\tilde{V}_{8}^{j}=2$, and at most one car at the point where $\tilde{V}_{22}^{j}=1$.
From the above, it can be said that the variable $\tilde{V}_{j}^{t}$ behaves like a traffic flow controller.
\subsection{Fundamental diagrams of the correlated Burgers cellular automata}
Let us consider the fundamental diagrams of the correlated Burgers cellular automata.
The density is denoted by
\begin{align*}
\rho =\dfrac{1}{NL} \sum_{j=0}^{N-1}U_{j}^{n},
\end{align*}
and the average flow $Q^{n}$ at the discrete time $n$ is defined as
\begin{align*}
Q^{n} &= \dfrac{1}{NL} \sum_{j=0}^{N-1}\min(V_{j+1}^{n-1},U_{j}^{n},L-U_{j+1}^{n})=\dfrac{1}{NL} \sum_{j=0}^{N-1}X_{j+1}^{n}.
\end{align*}
\par
Figures \ref{fig:kihonzu_L=1}, \ref{fig:kihonzu_L=2} and \ref{fig:kihonzu_L=3} show the fundamental diagrams in the case that $L=1,2$ and 3, respectively.
In these figures, we set $N=50$, and the initial values of $U_{j}^{0}$ and $\tilde{V}_{j}^{0}$ are given randomly by integer values between $[0,L]$ and $[1,L]$, respectively. 
Here we define $\tilde{V}_{\min}:=\min_{j=0,1,\dots,N-1}\tilde{V}_{j}^{0}$. 
Then, in these figures, $\tilde{V}_{\min}=1$.
Flows are computed at discrete time $n=100$ and are plotted in the figures.

\begin{figure}[t]
\begin{center}
\includegraphics[scale=0.5]{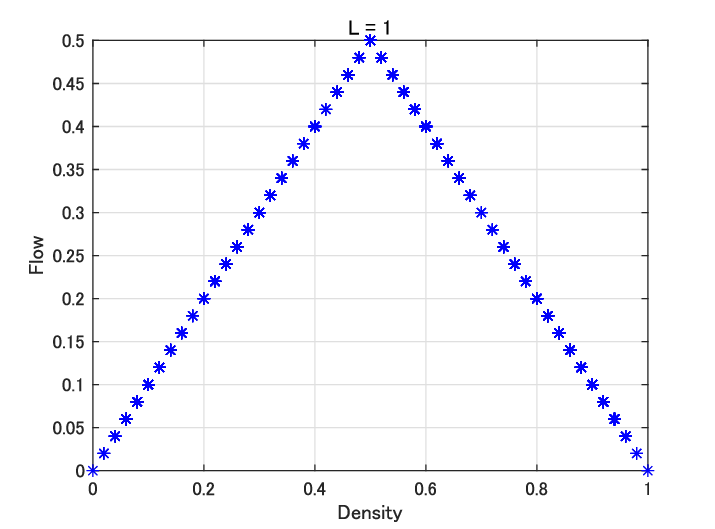}
\caption{Fundamental diagram of $L=1$. }\label{fig:kihonzu_L=1}
\end{center}
 \begin{minipage}{0.5\hsize}
  \begin{center}
   \includegraphics[scale=0.5]{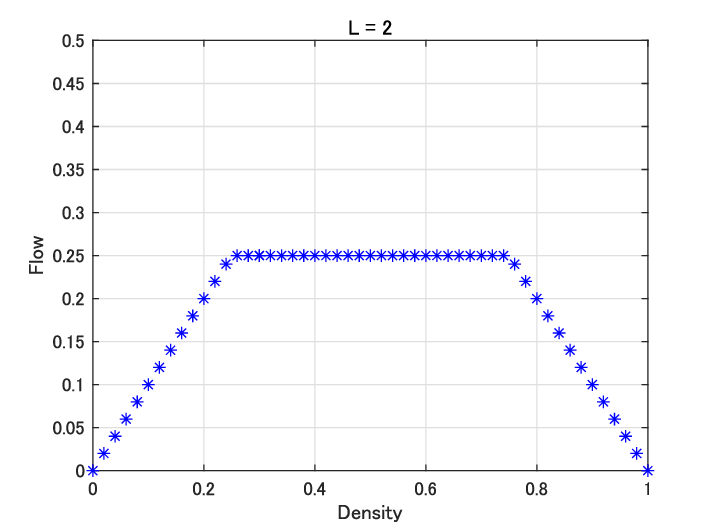}
  \end{center}
  \caption{Fundamental diagram of $L=2$.}
  \label{fig:kihonzu_L=2}
 \end{minipage}
 \begin{minipage}{0.5\hsize}
  \begin{center}
   \includegraphics[scale=0.5]{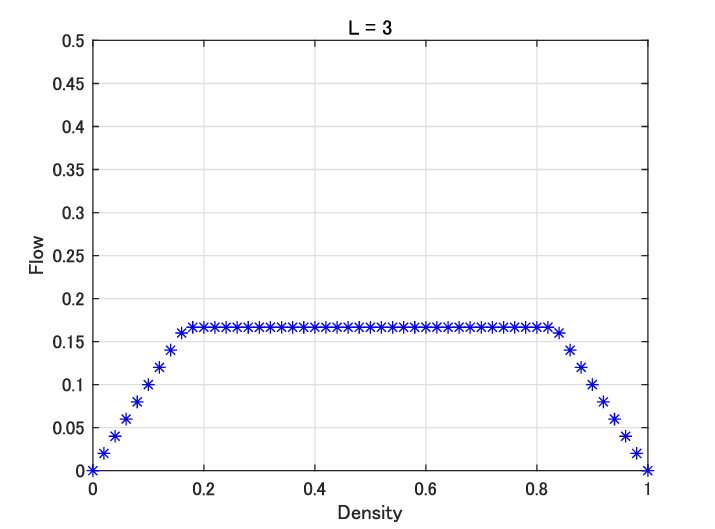}
  \end{center}
  \caption{Fundamental diagram of $L=3$.}
  \label{fig:kihonzu_L=3}
 \end{minipage}
\end{figure}

Figure \ref{fig:kihonzu_L=1} coincides with the fundamental diagram of the original traffic flow model of the elementary cellular automata of rule 184.
In Figures \ref{fig:kihonzu_L=2} and \ref{fig:kihonzu_L=3}, the values of the flow hit the ceiling and adopt a trapezoidal shape.
We can see the phase transitions at $\rho = 0.25$ and 0.75 in 
Figure \ref{fig:kihonzu_L=2} and $\rho = 0.1666...$ and 0.8333... in Figure \ref{fig:kihonzu_L=3}.

This result can be explained as follows.
Here we define the average velocity $v^{n}$ as
\begin{align}
v^{n}=\sum_{j=0}^{N-1}\dfrac{X_{j}^{n}}{U_{j}^{n}}.
\end{align}
The section in which the average velocity of cars is 1 is called free flow section and otherwise jam flow section.
In the following, we assume that after sufficient discrete time evolution, the whole $N$ sites will be either (i) only free flow section, (ii) mixed (free and jam) flow section, or (iii) only jam flow section. 
Let $v_{F}$ and $v_{J}$ be the average velocity in free and jam flow section, respectively. 
Let $\rho_{F}$ and $\rho_{J}$ be the density in free and jam flow section, respectively.  
Then, the density $\rho$ and the average flow $Q$ of the whole $N$ sites are given by
\begin{align}
&\rho = \rho_{J}x + \rho_{F}(1-x),\\
&Q = \rho_{J}v_{J}x + \rho_{F}v_{F}(1-x),\label{Q0}
\end{align}
where $x$ is the ratio of the length of the jam section to $N$. 
In the case where (i), since there are no jam flow section, $\rho=\rho_{F}$ and $v_{F}=1$ hold, then the average flow is $Q = \rho_{F}v_{F} = \rho_{F}=\rho$.
In the case where (ii), the right side of the site $p$ is free flow section and the left side is jam flow section, where $p \in \{0,1,\dots, N-1\}$ is a certain site which holds $\tilde{V}_{p}^{0}=\tilde{V}_{\min}$.
Then, the densities $\rho_{F}$ and $\rho_{J}$ are given by
\begin{align}
\rho_{F} = \dfrac{\tilde{V}_{\min}}{2L},
\quad \rho_{J} = \dfrac{2L-\tilde{V}_{\min}}{2L}=1-\rho_{F}.\label{rho_J}
\end{align}
The average velocities $v_{F}$, $v_{J}$ are given by
\begin{align}
v_{F} = \dfrac{\tilde{V}_{\min}}{\tilde{V}_{\min}} = 1,\quad 
v_{J} = \dfrac{\tilde{V}_{\min}}{2L-\tilde{V}_{\min}} = \dfrac{\rho_{F}}{1-\rho_{F}}.\label{V_J}
\end{align}
From \eqref{rho_J} and \eqref{V_J}, the average flow $Q$ can be written as
\begin{align}
Q = (1-\rho_{F})\dfrac{\rho_{F}}{1-\rho_{F}}x + \rho_{F}(1-x) = \rho_{F} = \dfrac{\tilde{V}_{\min}}{2L}
\end{align}
In the case where (iii), since there are no free flow section, $\rho=\rho_{J}$ holds and the average flow $Q$ is giben by
\begin{align}
Q=\rho_{J}v_{J} = \rho_{J}\dfrac{\tilde{V}_{\min}}{2L-\tilde{V}_{\min}} = 1-\rho_{J} =1-\rho.
\end{align}
Therefore, the fundamental diagram of the c-BCA is given by 
\begin{align}
Q = 
\left\{
\begin{array}{lcl}
\rho&\text{if}& 0\leq \rho \leq \dfrac{\tilde{V}_{\min}}{2L},\\
\dfrac{\tilde{V}_{\min}}{2L}&\text{if}& \dfrac{\tilde{V}_{\min}}{2L} \leq \rho \leq 1-\dfrac{\tilde{V}_{\min}}{2L},\\
1-\rho&\text{if}& 1-\dfrac{\tilde{V}_{\min}}{2L}\leq \rho \leq 1
\end{array}
\right.
\end{align}
Therefore, there exist two first-order phase transition of densities, and their points are given by $(\rho^{*},Q^{*})$ and $(1-\rho^{*},Q^{*})$, where 
\begin{align*}
\rho^{*} = Q^{*} = \dfrac{\tilde{V}_{\min}}{2L}.
\end{align*}
This result agrees with Figures \ref{fig:kihonzu_L=1}, \ref{fig:kihonzu_L=2} and \ref{fig:kihonzu_L=3}.

\section{Concluding remarks}
In this paper, we focus on the correlated random walk, which is a kind of generalization of the random walks, and we derive the discrete and ultradiscrete diffusion and Burgers equations and its corresponfing cellular automata. 
Thanks to the correlated random walk, our correlated Burgers cellular automata, the traffic flow model, behaves depending on the current and previous states. 
There are many studies on traffic flow models based on Burgers cellular automata and its extensions, as mentioned in Section 1. 
To the best of our knowledge, however, no studies have derived cellular automata model through a generalization of the random walk, which makes our present approach a highly original one.
In the study of the traffic flow model, it is important to analyze the fundamental diagram, and specifically, to identify the phase transition points.
We derive the fundermental diagram of the correlated Burgers cellular automata under the assumption of the convergence of dynamics. 
One of our future wors is to prove the convergence using the ultradiscrete correlated Burgers equation. 
The integrability of the conventional diffusion equation and the Burgers equation has already been studied, but that of our resulting equations has not yet been fully discussed, which will be another subject for a future work. 
A topic in continuous limit of the discrete correlated diffusion and discrete correlated Burgers will be written in a separate paper.

\section*{Acknowledgments}
The authors widh to thank Prof. Daichi Yanagisawa (The University of Tokyo) for valuable discussions.
This work was partially supported by the
Grant-in-Aid for Early-Career Scientists
no.20K14367 and
by the Grants-in-Aid for Scientific Research (C) nos.19K03616 and 19K03624 from the Japan Society for the Promotion of Science and the Research Origin for Dressed Photon.





\end{document}